\documentclass[journal,9pt]{IEEEtran}

\usepackage[utf8]{inputenc}
\usepackage{amssymb}
\usepackage{amsmath}
\usepackage{amsthm}
\usepackage{verbatim}
\usepackage{graphicx}

\usepackage[utf8]{inputenc}
\usepackage[english]{babel}

\newtheorem{theorem}{Theorem}[section]

\newtheorem{lemma}[theorem]{Lemma}

\begin{document}

\newpage

\title{Modified Dorfman procedure for pool tests with dilution - COVID-19 case study}
\author{Andrzej Jaszkiewicz 

\thanks{A. Jaszkiewicz is with Poznan University of Technology, Faculty of Computing, Institute of Computing Science, ul. Piotrowo 2, 60-965 Poznan, Poland, e-mail: andrzej.jaszkiewicz@put.poznan.pl. This research has been supported by the statutory funds of the Faculty of Computing and Telecommunications, Poznan University of Technology.

}}

\markboth{}
{Shell \MakeLowercase{\textit{et al.}}: Bare Demo of IEEEtran.cls for Journals}

\date{\today}

\maketitle

\begin{abstract}
The outbreak of the global COVID-19 pandemic results in unprecedented demand for fast and efficient testing of large numbers of patients for the presence of SARS-CoV-2 coronavirus. Beside technical improvements of the cost and speed of individual tests, pool testing may be used to improve efficiency and throughput of a population test. Dorfman pool testing procedure is one of the best known and studied methods of this kind. This procedure is, however, based on unrealistic assumptions that the pool test has perfect sensitivity and the only objective is to minimize the number of tests, and is not well adapted to the case of imperfect pool tests. We propose and analyze a simple modification of this procedure in which test of a pool with negative result is independently repeated up to several times. The proposed procedure is evaluated in a computational study using recent data about dilution effect for SARS‐CoV‐2 PCR tests, showing that the proposed approach significantly reduces the number of false negatives with a relatively small increase of the number of tests, especially for small prevalence rates. For example, for prevalence rate $0.001$ the number of tests could be reduced to $22.1\%$ of individual tests, increasing the expected number of false negatives by no more than $1\%$, and to $16.8\%$ of individual tests increasing the expected number of false negatives by no more than $10\%$. At the same time, a similar reduction of the expected number of tests in the standard Dorfman procedure would yield $675\%$ and $821\%$ increase of the expected number of false negatives, respectively. This makes the proposed procedure an interesting choice for screening tests in the case of diseases like COVID-19.
\end{abstract}


\section{Introduction}

Identification of infected patients is crucial for providing appropriate treatment and limit further transmission of COVID-19 disease. Large number of infected or potentially infected patients results in unprecedented demand for fast and efficient testing of large populations. For example, the results of a recent simulation study reported by Larremore et al. \cite{Larremore2020} demonstrate that effective screening depends largely on frequency of testing and the speed of reporting and could turn the COVID-19 tide within weeks. Some countries have decided already to test the whole population. For example, in October 2020 Slovakia decided to test all adults for SARS-CoV-2 \cite{Holt2020}. However, due to the cost, availability, and efficiency of PCR test, Slovakia's government decided to use antigen tests which are known to have lower sensitivity and specificity. 

Beside technical improvements of the cost and speed of individual tests, pool testing is an interesting option to improve efficiency and throughput of a population test.  Dorfman pool testing procedure is one of the best known and studied methods of this kind \cite{Dorfman1943}. In this procedure, the subjects are grouped into a number of non-overlapping pools and the specimens from the pool subjects are tested together. If the result of a test for a pool is negative, all subjects are assumed to be negative. Otherwise, all subjects from the pool are individually tested. Since, the prevalence rate is usually much below 50\%, many pools may contain only negative subjects and may be efficiently identified. 

Dorfman procedure procedure is, however, based on some assumptions that are not realistic in the context of COVID-19 and many other diseases, and thus limit its practical applications. This procedure assumes that the pool test has perfect sensitivity and the only objective is efficiency usually expressed as minimization of the number of tests. This assumption may seem justified in the case of PCR tests, including tests for presence of SARS‐CoV‐2 RNA, since in (reverse transcription) polymerase chain reaction the target DNA is exponentially multiplied, thus theoretically even extremely small amount of target DNA/RNA could be detected. In practice, the sensitivity of pool test may become lower due to the dilution effect \cite{McMahan2012,Stramer2013}. Indeed some recent studies confirm that dilution significantly affects sensitivity of SARS‐CoV‐2 PCR tests \cite{Abdalhamid2020,Bateman2020,Perchetti2020}. For example, Bateman et al. \cite{Bateman2020} estimated that pools of 5 (with one positive subjects) lead to 93\% sensitivity, pools of 10 lead to 91\% sensitivity, and pools of 50 lead to 81\% sensitivity, while estimated sensitivity for undiluted specimens is 99\%. Thus, dilution effects are significant and need to be appropriately taken into account.

With lower sensitivity of the pool test many subjects may be false negative. Such situations are highly undesired in the context of diseases like COVID-19 since they make impossible providing appropriate treatment and limiting further transmission. Thus, test efficiency should not be the only objective \cite{Malinovsky2016}. Although a number of authors proposed improvements of Dorfman testing procedure or other pool testing schemes, these modifications usually focus on improving efficiency (number of tests) \cite{Kim2007,PHATARFOD1994}, for example, by using a hierarchy of sub-pools tested in more than two stages. Hierarchical approach, however, even further increases the risk of a subject being false negative, since a subject is assumed negative it at least one of its (sub-)pool tests was negative. Other works take into account only reduced specificity with the assumption of perfect sensitivity \cite{AMOS2000, Gastwirth2000,Gupta1999}. There are relatively few studies taking into account accuracy related objectives with imperfect sensitivity in pool tests \cite{Aprahamian2018,Aprahamian2019,HWANG1975,Kim2007}. Still, some of the approaches make the assumption that the imperfect sensitivity is constant and the same for both individual and pool tests \cite{Aprahamian2019, Graff1972,Graff1974159,Kim2007,Nyongesa2004}, in other words, they do not take into account the dilution effect. Some works taking into account the dilution effect concern the problem of estimating the prevalence rate \cite{McMahan2013,Nguyen2018,Wang2015}, not the problem of identifying positive/negative subjects.  Furthermore, even the studies that consider dilution effect with false negatives as (partial) objective \cite{Aprahamian2018,HWANG1975,Wein1996} do not propose modifications of the Dorfman test procedure that could reduce the number of false negatives. In general, in our opinion, the development of efficient and precise pool testing procedures did not obtain sufficient interest from the research community, due to limited practical interest before the outbreak of COVID-19 pandemic. 

We show, that the standard Dorfman pool testing procedure is not well adapted to the case of imperfect pool tests with the objective of minimizing the number of false negatives and we propose a simple modification of this procedure in which test of a pool with negative result is independently repeated up to a given number of times. This approach significantly reduces the number of false negatives with a relatively small increase of the number of tests. 

We present also a numerical study based on the recently reported data about the influence of the dilution effect on the sensitivity of SARS‐CoV‐2 PCR tests \cite{Bateman2020}. We show that the modified procedure is better on both the expected number of false negatives and the number of tests than the standard Dorfman procedure for a wide range of parameters. We show also that the proposed procedure may significantly reduce the expected number of tests with practically negligible increase of the expected number of false negatives compared to individual tests.

Several authors proposed repeating test of pools (retesting) under some conditions \cite{Graff1972,Graff1974159, Nyongesa2004}. In particular our procedure could bee seen as the special case of the procedure proposed by Graff and Roeloffs \cite{Graff1972} with $r_1=1$ and $r_2=r$. Still there are some important differences between these works and our approach, namely:
\begin{itemize}
    \item These works do not take into account dilution effect, which significantly influences the analysis and performance of pool testing. 
    \item The presented analysis differs from ours, in particular we use Bayes' theorem to update posterior probabilities.
    \item We present a simulation study verifying the presented analysis.
    \item We do not propose to repeat test of pools with positive results, since although it could slightly reduce the number of false positives, it increases the number of tests. At the same time, the number of false positives has relatively low importance in the context of diseases like COVID-19.
\end{itemize}
Furthermore, in \cite{Nyongesa2011} Nyongesa described a procedure in which a pool with negative result is retested just once, however, this procedure was analyzed only from the point of view of prevalence rate estimation without dilution effect.

The contributions of this paper are:
\begin{itemize}
    \item A modified Dorfman pool testing procedure for pool tests with dilution effect designed to improve the number of false negatives at a moderate cost of increased number of tests is proposed.
    \item Analytical formulas for the number of tests, the number of false negatives and the number of false positives in the proposed procedure are derived.
    \item The analytical formulas are verified in a computational simulation study.
    \item The practical contribution of the paper is that through a computational case study using recent data about SARS‐CoV‐2 PCR tests sensitivity it is shown that the proposed procedure may be highly beneficial in the case of diseases like COVID-19.
\end{itemize}

The paper is organized in the following way. In the next Section, we define basic concepts and notation. Then, we present the analysis of the standard Dorfman pool testing procedure from the point of view of the number of tests, the number of false negatives and the number of false positives. In Section four, we introduce the proposed procedure and analyze its properties. In the fifth Section, we present a computational case study for SARS‐CoV‐2 PCR tests. Finally, we present conclusions and directions for further research.

\section{Preliminaries}

Let $p = P(P_s=1)$ be the probability of a subject being positive (prevalence rate), where $P_s \in \{0,1\}$ is a random variable describing that the subjects is positive ($P_s=1$) or negative ($P_s=0$). Let
$Sp$ be the test specificity, i.e. the conditional probability that the outcome of the test is negative given the subjects is negative. Let
$Se$ be test sensitivity, i.e. the conditional probability that the outcome of the test is positive given the subjects is positive.

Probability of a test being true negative in a given test is $Sp\,(1-p)$.

Probability of a test being true positive in a given test is $Se\,p$.

Probability of a test being false negative in a given test is $(1-Se)p$.

Probability of a test being false positive in a given test is $(1-Sp)(1-p)$.

Let $Sp_I$ and $Se_I$ denote specificity and sensitivity of an individual test. 
If each subject is tested individually, the expected number of false negatives $\mathbb{E}(FN)$, the expected number of false positives $\mathbb{E}(FP)$, and the number of tests $T$ per subject are:

\begin{equation}
\mathbb{E}(FN)=(1-Se)p \\
\end{equation}
\begin{equation}
\mathbb{E}(FP)=(1-Sp)(1-p) \\
\end{equation}
\begin{equation}
T=1
\end{equation}

Consider a pool of size $n$. The probability that there is at least one positive subject in the pool (pool is positive) is 
\begin{equation}
p_P=1-(1-p)^n
\end{equation}

Let $Sp_P$ and $Se_P$ denote specificity and sensitivity of a pool test. It is reasonable to assume that $Sp_P=Sp_I=Sp$, since in both cases we deal with a sample containing no virus RNA. 
\begin{equation}
\begin{split}
Se_P=\mathbb{E}\big(Se_P(n,k), k=1,\dots,n\big)= \\
\frac{\sum_{k=1}^{k=n}\big(Se_P(n,k)Pr(k;n,p)\big)}{\sum_{k=1}^{k=n}Pr(k;n,p)}=\\
\frac{\sum_{k=1}^{k=n}\big(Se_P(n,k)Pr(k;n,p)\big)}{p_P}
\end{split}
\end{equation}
where $Se_P(n,k)$ is the sensitivity of a pool test with the pool of size $n$ and $k$ positive subjects in the pool, and $Pr(k;n,p)$ is the probability of having $k$ positive subjects in the pool of size $n$ that could be derived from binomial distribution. $Se_P$ is the same for each pool assuming constant pool size and homogeneous subjects (i.e. constant $p$). In practice $Sp_P=Sp_I=Sp$ and $Se_I$ are close to $1$. $Se_P, Se_P(n, k) \leq Se_I$ and $Se_P(n, k)$ decreases with growing $n$ and increases with $k$.

Let $T \in \{0,1\}$ denote the random variable being outcome of the pool test. The probability of a pool test being (true or false) positive is:
\begin{equation}
P(T=1)=Se_P\,p_P+(1-Sp)(1-p_P)
\end{equation}
and the probability of a pool test being (true or false) negative is:
\begin{equation}
P(T=0)=(1-Se_P)p_P+Sp(1-p_P)
\end{equation}

\section{Standard Dorfman pool testing procedure}

The detailed analysis is presented in Appendix \ref{App:A}. In this Section we present the analytical formulas for the expected number of false negatives $\mathbb{E}(FN')$, number of false positives $\mathbb{E}(FP')$, and expected number of tests $\mathbb{E}(T')$.
\begin{equation}
\begin{split}
\mathbb{E}(FN')=\\
p\sum_{k=1}^{k=n}\Big(\big((1 - Se_ISe_P(n,k))Pr(k-1;n-1,p)\Big) 
\end{split}
\end{equation}
\begin{equation}
\begin{split}
\mathbb{E}(FP')= \\
(1-Sp)\Big(P(T=1)- \\ \sum_{k=1}^{k=n}\big(Se_P(n,k)Pr(k-1;n-1,p)\big)p\Big)
\end{split}
\end{equation}
\begin{equation}
\mathbb{E}(T')=\frac{1}{n}+P(T=1)
\end{equation}

\section{Modified pool testing procedure}

The weak point of pool tests is decreased sensitivity  compared to an individual test. In result, the number of false negatives may be increased. This weak point could be, however, resolved by repeating test for the same pool, if the results of the previous tests were negative, until the last test was positive or a given number of tests were negative. In other words, we assume that the pool test is positive, if at least one result was positive and we assume that the pool test is negative, if all results were negative. The decision tree of the repeated pool test is presented in Figure \ref{fig:scheme}. Note, that we assume that for each pool test, different samples from each subject are used, so the results of the repeated tests may be assumed to be independent. 

\begin{figure}
\begin{center}
  \includegraphics[width=0.48\textwidth]{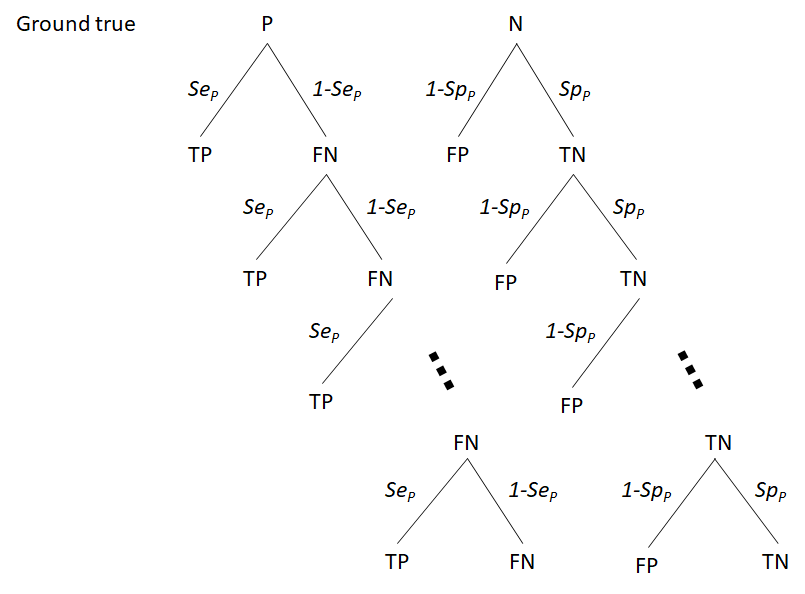}
\caption{Decision tree of the repeated pool test}
\label{fig:scheme}       
    
\end{center}
\end{figure}

Let $T'' \in \{0,1\}$ denote the random variable being outcome of the pool test (with potential repeats). The probability of a pool test being (true or false) positive is:
\begin{equation}
P(T''=1)=Se''_P\,p_P+(1-Sp'')(1-p_P)
\end{equation}
and the probability of a pool test being (true or false) negative is:
\begin{equation}
P(T''=0)=(1-Se''_P)p_P+Sp''(1-p_P)
\end{equation}
where:
\begin{equation}
Sp''_P=Sp^r \leq Sp
\end{equation}
is the repeated test specificity and  
\begin{equation}
\begin{split}
Se_P''(n,k)= 
1-\big(1-Se_P(n,k)\big)^r
\geq \\
Se_P(n,k), k=1,\dots,n 
\end{split}
\end{equation}
\begin{equation}
Se_P''=\mathbb{E}\big(Se_P''(n,k), k=1,\dots,n\big) \geq Se_P
\end{equation}
is the repeated test sensitivity (see Appendix \ref{App:B}).  

\begin{lemma}
$P(T''=1) \geq P(T=1)$ and $P(T''=0) \leq P(T=1)$.
\end{lemma}
\begin{proof}
It is direct consequence of $Sp''_P \leq Sp$ and $Se_P'' \geq Se_P$
\end{proof}

The detailed analysis is presented in Appendix \ref{App:B}. In this Section we present the analytical formulas for the expected number of false negatives $\mathbb{E}(FN'')$, expected number of false positives $\mathbb{E}(FP'')$, the expected number of individual tests $\mathbb{E}(T''_I)$, and expected number of tests $\mathbb{E}(T'')$.

\begin{equation}
\begin{split}
\mathbb{E}(FN'')= \\
\,p\sum_{k=1}^{k=n}\Big(\big((1-Se_I\,Se''_P(n,k))Pr(k-1;n-1,p)\Big)
\end{split}
\end{equation}
\begin{equation}
\begin{split}
\mathbb{E}(FP'')= \\
(1-Sp)\Big(P(T''=1) - \\
\sum_{k=1}^{k=n}\big(Se''_P(n,k)Pr(k-1;n-1,p)\big)p\Big)
    \end{split}
\end{equation}
\begin{equation}
\begin{split}
\mathbb{E}(T''_I)=P(T''=1) \\
\end{split}
\end{equation}
\begin{equation}
\begin{split}
\mathbb{E}(T'')= \\
\frac{1}{n}+ \\
\frac{1}{n}\sum_{l=2}^{r}\Big((1-p_P)Sp^{l-1} + \\  \sum_{k=1}^{k=n}\big(1-Se_P(n, k)\big)^{l-1}Pr(k;n,p)\Big)+\\
P(T''=1)
\end{split}
\end{equation}

\begin{lemma}
With the same parameters, the modified procedure has lower or equal expected number of false negatives than the standard Dorfman test:
$$\mathbb{E}(FN'') \leq \mathbb{E}(FN') 
$$
\end{lemma}
\begin{proof}
It is a direct consequence of $Se_P''(n,k) \geq Se_P(n,k), k=1,\dots,n$.
\end{proof}

\begin{lemma}
With the same parameters, the modified procedure has greater or equal expected number of individual tests and total number of tests than the standard Dorfman test:
$$\mathbb{E}(T''_I) \geq \mathbb{E}(T'_I)$$
$$\mathbb{E}(T'') \geq \mathbb{E}(T')$$
\end{lemma}
\begin{proof}
$\mathbb{E}(T''_I) \geq \mathbb{E}(T'_I)$ results from $P(T''=1) \geq P(T=1)$.
$\mathbb{E}(T'') \geq \mathbb{E}(T')$ results from $\mathbb{E}(T''_I) \geq \mathbb{E}(T'_I)$ and the fact that already the first iteration of pool tests in the modified procedure requires the same number of tests as the first stage in the standard procedure.
\end{proof}

\begin{lemma}
The expected number of false negatives if greater or equal than in the case of individual tests:
$$\mathbb{E}(FN'') \geq \mathbb{E}(FN)$$
and $\mathbb{E}(FN'')$ achieves minimum equal to $(1-Se_I)p = \mathbb{E}(FN)$, if $Se''_P(n,k) = 1, k=1,\dots,n$. 
\end{lemma}
\begin{proof}
$Se''_P(n,k)) \in [0,1], k=1,\dots,n$ and $\mathbb{E}(FN'')$ in non-increasing with $Se''_P(n,k)), k=1,\dots,n$, thus $\mathbb{E}(FN'')$ achieves minimum for $Se''_P(n,k) = 1, k=1,\dots,n$.

If $Se''_P(n,k) = 1, k=1,\dots,n$ then:
\begin{equation}
\begin{split}
\mathbb{E}(FN'')= \\
\,p\,\sum_{k=1}^{k=n}\Big(\big((1-Se_I)Pr(k-1;n-1,p)\Big) = \\ \,p(1-Se_I)\sum_{k=1}^{k=n}Pr(k-1;n-1,p) = \\
\,p\,(1-Se_I) = \mathbb{E}(FN)
\end{split}
\end{equation}
\end{proof}

Note, that since values of $Se_P''(n,k)$ become very close to 1 in the repeated pool test, $\mathbb{E}(FN'')$ may be very close to $\mathbb{E}(FN)$.


\section{Computational case study}

We verify the presented analysis and evaluate the proposed procedure using recent data about the influence of the dilution effect on the sensitivity of SARS‐CoV‐2 PCR tests \cite{Bateman2020}. Since, however, the authors presented results (values of $Se(n,k)$) only for some specific values of $n=1,5,10,50$ and $k=1$ we need to fit to the data a model allowing calculation of $Se(n,k)$ for other values of $n$ and $k$. We tried first the model proposed in \cite{Burns1987} and used also in \cite{Aprahamian2018}:
\begin{equation}
Se(n,k)=1-Sp+(Se_I+Sp-1)(\frac{n}{k})^\alpha
\end{equation}
with $Sp \geq 0.99$, however, this model did not fit the data perfectly (MSE = $0.000781$). So, we decided to extend this model with a linear term:
\begin{equation}
Se(n,k)=1-Sp+(Se_I+Sp-1)(\frac{n}{k})^\alpha+\beta\,k
\end{equation}
This model fitted the data almost perfectly (MSE = $0.000026$) with $\alpha=0.032482$ and $\beta=-0.001255$. We hypothesize that the linear term reflects PCR test dependence not only on the relative but also on the absolute amount of virus RNA and remaining material.

We tested the proposed procedure for prevalence rates $p = 0.001$, $0.002$, $0.005$, $0.01$, $0.02$, $0.05$, $0.1$, $0.2$, $0.3$. Our procedure involves two parameters, pool size $n$ and number of iterations of pool test $r$. Since their values are integer and reasonable ranges of their values are relatively small, all potential settings could be enumerated. We used $n=2,\dots,50$ and $r=2,\dots,5$. We did not try higher numbers of iterations, since they would be highly impractical and $5$ iterations was enough to assure the expected number of false negatives very close to that of individual tests. 

We performed two kinds of computational studies. First, we used the analytical formulas presented above. Second, we simulated the procedures for $100,000,000$ subjects drawn randomly as either positive or negative with a given prevalence rate. By comparing results of the analytical and simulation studies we verified the analytical formulas. In Table \ref{tab:ver} we present mean relative squared errors between the analytical and simulation results for particular parameters and methods. These results confirm very good agreement between the analytical and simulation study.

We compare the proposed procedure to the standard Dorfman procedure and to individual tests. Since false positives are much less important than false negatives in the context of diseases like COVID-19, we decided to use two main objectives i.e. the expected number of false negatives and the expected number of tests. Thus, we deal with a bi-objective optimization problem with two integer variables defining feasible solutions (settings of parameters) \cite{Handl2007}. For each value of $p$, only Pareto-optimal solutions were preserved. A solution is Pareto-optimal (non-dominated, efficient) if there exists no other feasible solution better on one of the objectives and not worse on other objective(s). For solutions that are not Pareto-optimal it is possible to improve one objective without deteriorating other(s), so they do not constitute reasonable parameters for pool tests.

\begin{table}[t]
\caption{Mean relative squared errors between the analytical and simulated results}
\begin{center}
\label{tab:ver}
\begin{tabular}{cccc}
Parameter\textbackslash method & Individual  & Dorfman & Proposed \\ \hline
$\mathbb{E}(T)$ &  & $1.20\mathrm{e}{-6}$ & $6.91\mathrm{e}{-7}$ \\
$\mathbb{E}(FN)$ & $1.53\mathrm{e}{-4}$ & $2.10\mathrm{e}{-5}$ & $1.78\mathrm{e}{-4}$ \\
$\mathbb{E}(FP)$ & $3.78\mathrm{e}{-7}$ & $1.33\mathrm{e}{-5}$ & $6.19\mathrm{e}{-5}$ \\

\end{tabular}
\end{center}
\end{table}

\begin{table}[t]
\caption{Minimum relative expected number of tests assuring a given maximum increase of the expected number of false negatives}
\begin{center}
\label{tab:p}
\begin{tabular}{cccc}
$p$\textbackslash $\mathbb{E}(FN)$ increase & $\leq 1\%$  & $\leq 10\%$ & $\leq 100\%$ \\ \hline
0.001 & $22.1\%, r=5$ & $16.6\%, r=4$ & $13.4\%,r=3$ \\
0.002 & $24.9\%, r=5$ & $20.5\%, r=4$ & $17.5\%,r=3$ \\
0.005 & $32.5\%, r=5$ & $29.2\%, r=4$ & $23.8\%,r=2$ \\
0.01 & $42.3\%, r=5$ & $37.0\%, r=3$ & $28.2\%,r=2$ \\
0.02 & $52.4\%, r=4$ & $45.3\%, r=3$ & $37.8\%,r=2$ \\
0.05 & $70.1\%, r=4$ & $64.1\%, r=3$ & $55.8\%,r=2$ \\
0.1 & $87.1\%, r=4$ & $81.9\%, r=3$ & $73.7\%,r=2$ \\
0.2 & - & - & - \\
0.3 & - & - & - \\

\end{tabular}
\end{center}
\end{table}

\begin{table}[t]
\caption{Average (expected) numbers of false positives}
\begin{center}
\label{tab:FP}
\begin{tabular}{cccc}
$p$& Individual tests  & Standard Dorfman & Proposed \\ \hline
0.001	&	0.00999	&	0.00039	&	0.00065	\\
0.002	&	0.00998	&	0.00051	&	0.00082	\\
0.005	&	0.00995	&	0.00073	&	0.00117	\\
0.01	&	0.00990	&	0.00094	&	0.00147	\\
0.02	&	0.00980	&	0.00127	&	0.00195	\\
0.05	&	0.00950	&	0.00173	&	0.00300	\\
0.1	&	0.00900	&	0.00237	&	0.00379	\\
0.2	&	0.00800	&	0.00280	&	0.00479	\\
0.3	&	0.00700	&	0.00622	&	0.00698	\\

\end{tabular}
\end{center}
\end{table}

\begin{figure}
\begin{tabular}{c}
 \includegraphics[width=0.4\textwidth]{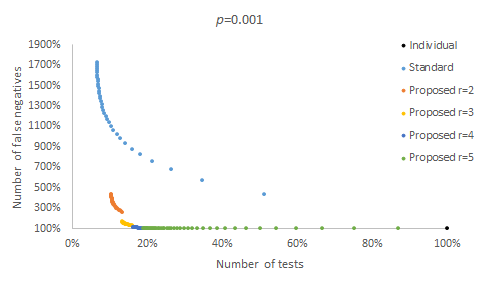} \\   \includegraphics[width=0.4\textwidth]{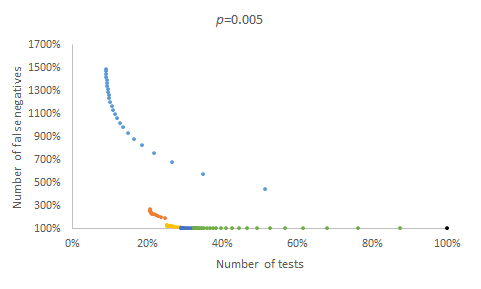} \\
 \includegraphics[width=0.4\textwidth]{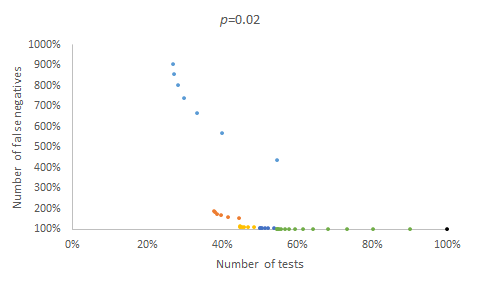} \\   \includegraphics[width=0.4\textwidth]{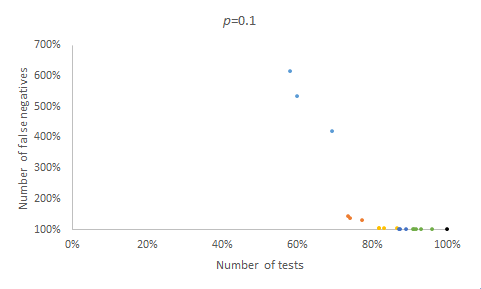} \\
\end{tabular}
\caption{Pareto optimal solutions for compared procedures}
\label{fig:Pareto}
\end{figure}

All detailed numerical results and the source code used in this computational study are available at dedicated web page \footnote{https://sites.google.com/view/pooltests/main}. In Figure \ref{fig:Pareto} we present exemplary sets of Pareto-optimal solutions of the standard Dorfman and the proposed procedure for selected values of $p$. The results are presented relative to the results corresponding to individual tests of all subjects. As it can be seen, the proposed procedure performs especially well for low prevalence rates and is able to assure much lower expected numbers of false negatives than the standard Dorfman procedure with the similar expected number of tests. Note, however, that if reduction of the number of tests is the main objective, the standard Dorfman testing procedure could assure lower expected number of tests than the proposed procedure.

Furthermore, the proposed procedure allows for a significant reduction of the expected number of tests with very low increase of the expected number of false negatives compared to individual tests. To analyze it more precisely, in Table \ref{tab:p} we present Pareto-optimal solutions for the proposed procedure that assure no more than $1\%$, $10\%$, and $100\%$ increase of the expected number of false negatives with the minimum expected number of tests. Note, that for $p=0.2$ and $p=0.3$ no such solutions exist with the expected number of tests lower than the number of tests in individual testing. These results confirm that the proposed procedure may be highly beneficial with low prevalence rates. For example, for $p=0.001$ the number of tests could be reduced to $22.1\%$ of individual tests increasing the expected number of false negatives by no more than $1\%$, and to $16.8\%$ of individual tests increasing the expected number of false negatives by no more than $10\%$. At the same time, a similar reduction of the expected number of tests (to $26.4\%$ and $18.2\%$, respectively) in the standard Dorfman procedure would yield $675\%$ and $821\%$ increase of the expected number of false negatives, respectively.

Although, we treat the numbers of tests and of false negatives as the main objectives, in Table \ref{tab:FP} we report also the average numbers of false positives for Pareto-optimal solutions in compared procedures. As it could be observed, the number of false positives in the proposed procedure is larger than in the standard Dorfman procedure, but is lower than in individual tests.

\section{Conclusions and directions for further research}

We have presented a modified pool testing procedure which could significantly reduce the number of tests at a very low cost of increasing the number of false negatives. We have derived analytical formulas for the main objectives and verified them iscreening n a simulation study. We have also shown that the proposed procedure may be highly beneficial for SARS‐CoV‐2 PCR screening tests.

In this work we assumed a homogeneous population of subjects with constant prevalence rate. An interesting research direction would be to consider heterogeneous subjects with different prevalence rates \cite{Aprahamian2018}. Another possibility would be to apply hierarchical approach \cite{Kim2007} in the case of pools tested positively that could perhaps further reduce the expected number of tests.

In this study, we assumed that pool test is repeated with the same pool. Another possibility would to re-shuffle the subjects from the negative pools into new pools. It is yet to be established is such approach could bring some benefits.

Note, however, that adding a hierarchical approach or pools re-shuffling would increase complexity, while an advantage of the proposed procedure is relative simplicity which may be very important for practical implementations.

\begin{appendices}

\section{Standard Dorfman pool testing procedure}
\label{App:A}

After obtaining a pool test result, the posterior probability of being positive should be updated for the subjects in this pool using Bayes' theorem. If the test result was negative i.e. $T=0$:
\begin{equation}
\begin{split}
P(P_s=1|T=0)=\\
\frac{P(T=0|P_s=1)P(P_s=1)}{P(T=0)}=\frac{P(T=0|P_s=1)p}{P(T=0)}
\end{split}
\end{equation}

\begin{equation}
\begin{split}
P(T=0|P_s=1)= \\
\mathbb{E}\big(1-Se_P(n,k), k=1,\dots,n|P_s=1\big)= \\
\frac{\sum_{k=1}^{k=n}\Big(\big(1-Se_P(n,k)\big)Pr(k;n,p|P_s=1)\Big)}{\sum_{k=1}^{k=n}Pr(k;n,p|P_s=1)}= \\
\sum_{k=1}^{k=n}\Big(\big(1-Se_P(n,k)\big)Pr(k;n,p|P_s=1)\Big)
\end{split}
\end{equation}

Having $k$ positive subjects given $P_s=1$ means that there are $k-1$ positive subjects among $n-1$ remaining subjects, thus $Pr(k;n,p|P_s=1) = Pr(k-1;n-1,p)$, so:
\begin{equation}
P(T=0|P_s=1)=\sum_{k=1}^{k=n}\Big(\big(1-Se_P(n,k)\big)Pr(k-1;n-1,p)\Big)
\end{equation}
Thus, the posterior probability is:
\begin{equation}
\begin{split}
(P_s=1|T=0)= \\
\frac{\sum_{k=1}^{k=n}\Big(\big(1-Se_P(n,k)\big)Pr(k-1;n-1,p)\Big)\,p}{P(T=0)}
\end{split}
\end{equation}

In an analogous way we may analyse the case $T=1$: 
\begin{equation}
\begin{split}
P(P_s=1|T=1)= \\
\frac{P(T=1|P_s=1)p}{P(T=1)}=\frac{P(T=1|P_s=1)p}{P(T=1)}
\end{split}
\end{equation}
\begin{equation}
\begin{split}
P(T=1|P_s=1)= \\
\mathbb{E}\big(Se_P(n,k), k=1,\dots,n\big|P_s=1)= \\
\sum_{k=1}^{k=n}\big(Se_P(n,k)Pr(k;n,p|P_s=1)\big)= \\
\sum_{k=1}^{k=n}\big(Se_P(n,k)Pr(k-1;n-1,p)\big)
\end{split}
\end{equation}
Thus:
\begin{equation}
\begin{split}
P(P_s=1|T=1)= \\
\frac{\sum_{k=1}^{k=n}\big(Se_P(n,k)Pr(k-1;n-1,p)\big)p}{P(T=1)}
\end{split}
\end{equation}


The expected number of subjects in negative pools that are not individually tested is:
\begin{equation}
P(T=0)
\end{equation}
The probability that one of these subjects is false negative is $P(P_s=1|T=0)$
Thus, the expected number of false negative subjects from the pool testing stage is:
\begin{equation}
\begin{split}
\mathbb{E}(FN'_P)=P(P_s=1|T=0)P(T=0)= \\
\frac{\sum_{k=1}^{k=n}\Big(\big((1-Se_P(n,k))Pr(k-1;n-1,p)\big)\Big)\,p\,P(T=0)}{P(T=0)}= \\
\sum_{k=1}^{k=n}\Big(\big((1-Se_P(n,k))Pr(k-1;n-1,p)\Big)\,p
\end{split}
\end{equation}


A subject is individually tested if its pool test result was true or false positive. The expected number of subjects individually tested in the second phase is:
\begin{equation}
\mathbb{E}(T'_I)=P(T=1)
\end{equation}
So, the expected number of false negatives in the second stage is:
\begin{equation}
\begin{split}
\mathbb{E}(FN'_I)=(1-Se_I)P(P_s=1|T=1)P(T=1)= \\
(1-Se_I)\sum_{k=1}^{k=n}\big(Se_P(n,k)Pr(k-1;n-1,p)\big)p
\end{split}
\end{equation}

In total the expected number of false negative subjects in both stages is:
\begin{equation}
\begin{split}
\mathbb{E}(FN')=\mathbb{E}(FN'_P)+\mathbb{E}(FN'_I)= \\
\sum_{k=1}^{k=n}\Big(\big((1-Se_P(n,k))Pr(k-1;n-1,p)\Big)\,p+ \\
(1-Se_I)\sum_{k=1}^{k=n}\big(Se_P(n,k)Pr(k-1;n-1,p)\big)p= \\
p\sum_{k=1}^{k=n}\Big(\big((1-Se_P(n,k) + \\ (1-Se_I)Se_P(n,k))Pr(k-1;n-1,p)\Big)= \\
p\sum_{k=1}^{k=n}\Big(\big((1 - Se_ISe_P(n,k))Pr(k-1;n-1,p)\Big)
\end{split}
\end{equation}

A subject is false positive, if its pool test was true or false positive and the individual test was false positive. Thus, the expected number of false positive subjects is:
\begin{equation}
\begin{split}
\mathbb{E}(FP')=\mathbb{E}(T'_I)(1-Sp)(1-P(P_s=1|T=1))= \\
P(T=1)(1-Sp)(1- \\
\frac{\sum_{k=1}^{k=n}\big(Se_P(n,k)Pr(k-1;n-1,p)\big)p}{P(T=1)})= \\
(1-Sp)\Big(P(T=1) - \\
\sum_{k=1}^{k=n}\big(Se_P(n,k)Pr(k-1;n-1,p)\big)p\Big)
\end{split}
\end{equation}

To calculate the expected number of tests we add the number of pool tests and the expected number of subjects tested individually in the second phase:
\begin{equation}
\mathbb{E}(T')=\frac{1}{n}+\mathbb{E'}(T_I)=\frac{1}{n}+P(T=1)
\end{equation}

\section{Modified pool testing procedure}

\label{App:B}

In repeated pool test we need to update test specificity and sensitivity for each value of $k$. As can be seen from the Figure \ref{fig:scheme} the repeated pool test specificity with $r$ iterations is 
\begin{equation}
Sp''_P=Sp^r \leq Sp
\end{equation}
and the repeated pool test sensitivity is  
\begin{equation}
\begin{split}
Se_P''(n,k)= 
1-\big(1-Se_P(n,k)\big)^r
\geq \\
Se_P(n,k), k=1,\dots,n 
\end{split}
\end{equation}
Obviously $Se''_P(n,k) > Se_P(n,k), k=1,\dots,n$, if $Se_P(n,k) < 1$. 

Note that the sensitivity of the repeated pool test may be higher than the sensitivity of an individual test if
\begin{equation}
Se''_P=\frac{\sum_{k=1}^{k=n}\big(Se''_P(n,k)Pr(k;n,p)\big)}{p_P}>Se_I
\end{equation}

Following the analysis presented for the standard Dorfman pool testing procedure we can calculate the expected numbers of false negatives in the first $\mathbb{E}(FN''_P)$ and the second $\mathbb{E}(FN''_I)$ stage, the total expected numbers of false negatives $\mathbb{E}(FN''_I)$, the expected number of tests in the second stage $\mathbb{E}(T''_I)$, and the expected number of false positives $\mathbb{E}(FP'')$:

\begin{equation}
\mathbb{E}(T''_I)=P(T''=1)
\end{equation}

\begin{equation}
\begin{split}
\mathbb{E}(FN''_P)= \\
\sum_{k=1}^{k=n}\Big(\big((1-Se''_P(n,k))Pr(k-1;n-1,p)\Big)\,p
\end{split}
\end{equation}

\begin{equation}
\begin{split}
\mathbb{E}(FN''_I)=\\
(1-Se_I)\sum_{k=1}^{k=n}\big(Se''_P(n,k)Pr(k-1;n-1,p)\big)p
\end{split}
\end{equation}

\begin{equation}
\begin{split}
\mathbb{E}(FN'')= \\
\,p\sum_{k=1}^{k=n}\Big(\big((1-Se_I\,Se''_P(n,k))Pr(k-1;n-1,p)\Big)
\end{split}
\end{equation}
\begin{equation}
\begin{split}
\mathbb{E}(FP'')= \\
(1-Sp)\Big(P(T''=1) -\sum_{k=1}^{k=n}\big(Se''_P(n,k)Pr(k-1;n-1,p)\big)p\Big)
    \end{split}
\end{equation}



To calculate the expected number of tests we need to take into account that the pool test for a pool is repeated if all previous tests for this pool were negative. Consider iteration number $l$. The probability that a pool was true negative in iterations $1,\dots,l-1$ is:
\begin{equation}
(1-p_P)Sp^{l-1}
\end{equation}
To calculate the probability that a pool was false negative in iterations $1,\dots,l-1$ we need to consider each $k=1,\dots,n$. This probability for a given $k$ is:
\begin{equation}
p_P(1-Se_P(n, k))^{l-1}
\end{equation}
and for all $k=1,\dots,n$:

\begin{equation}
\begin{split}
p_P\frac{\sum_{k=1}^{k=n}\Big(\big(1-Se_P(n, k)\big)^{l-1}Pr(k;n,p)\Big)}{p_P}= \\ \sum_{k=1}^{k=n}\Big(\big(1-Se_P(n, k)\big)^{l-1}Pr(k;n,p)\Big)
\end{split}
\end{equation}
thus the probability that a pool was negative in iterations $1,\dots,l-1$ is:
\begin{equation}
(1-p_P)Sp^{l-1} + \sum_{k=1}^{k=n}\Big(\big(1-Se_P(n, k)\big)^{l-1}Pr(k;n,p)\Big)
\end{equation}
and the expected number of tests is
\begin{equation}
\begin{split}
\mathbb{E}(T'')= \\
\frac{1}{n}+ \\
\frac{1}{n}\sum_{l=2}^{r}\Big((1-p_P)Sp^{l-1} + \sum_{k=1}^{k=n}\big(1-Se_P(n, k)\big)^{l-1}Pr(k;n,p)\Big)+ \\
\mathbb{E}(T''_I)= \\
\frac{1}{n}+ \\
\frac{1}{n}\sum_{l=2}^{r}\Big((1-p_P)Sp^{l-1} + \sum_{k=1}^{k=n}\big(1-Se_P(n, k)\big)^{l-1}Pr(k;n,p)\Big)+ \\
P(T''=1)
\end{split}
\end{equation}

\end{appendices}

\bibliographystyle{plain}
\bibliography{references}
\end{document}